\documentclass[twocolumn,showpacs,superscriptaddress,10pt]{revtex4-1} 
\usepackage{amsmath,amssymb,graphicx,amsthm}

\newtheorem{theorem}{Theorem}

\newtheorem{proposition}{Proposition}

\newtheorem{definition}[theorem]{Definition}

\newcommand*{\cA}{\mathcal{A}} 

\newcommand*{\cC}{\mathcal{C}}

\newcommand*{\cI}{\mathcal{I}}

\newcommand*{\cM}{\mathcal{M}}

\newcommand*{\id}{\mathsf{id}}
\newcommand*{\ket}[1]{|#1\rangle}
\newcommand*{\bra}[1]{\langle #1|}

\newcommand*{\proj}[1]{\ket{#1}\bra{#1}}
\newcommand*{\Tr}{\mathsf{Tr}}

\newcommand{\beq}{\begin{equation}}
\newcommand{\eeq}{\end{equation}}

\newcommand{\channel}[2]{\Lambda^{#1\rightarrow #2}}

\begin{document}

\title{Channel Steering}

\author{Marco Piani}
\affiliation{SUPA  and Department of Physics, University of Strathclyde, Glasgow G4 0NG, UK}
\affiliation{Institute for Quantum Computing and Department of Physics and Astronomy, University of Waterloo, N2L 3G1 Waterloo ON, Canada}

\begin{abstract} We introduce and study the notion of steerability for channels. This generalizes the notion of steerability of bipartite quantum states. We discuss a key conceptual difference between the case of states and the case of channels:  while state steering deals with the notion of ``hidden'' states, steerability in the channel case is better understood in terms of coherence of channel extensions, rather than in terms of ``hidden'' channels. This distinction vanishes in the case of states. We further argue how the proposed notion of lack of coherence of channel extensions coincides with the notion of channel extensions realized via local operations and classical communication. We also discuss how the Choi-Jamio{\l}kowski isomorphism allows the direct application of many results about states to the case of channels. We introduce measures for the steerability of channel extensions.
\end{abstract}


\maketitle 

\section{Introduction}
The phenomenon of quantum steering has attracted much attention recently, both from the theoretical and experimental perspective (see, e.g.,~\cite{PhysRevLett.113.160403,PhysRevLett.113.160402,PhysRevLett.113.140402} and references therein).  Already identified in the early days of quantum mechanics~\cite{einstein1935can,schroedinger},  but put on solid theoretical grounds only recently~\cite{wisemanPRL2007,jonesPRA2007}, steering is a property of states of bipartite systems regarding the fact that one party, through local measurements, can realize ensembles for the reduced state of the other party that do not admit an explanation in terms of classical correlations, and in particular in terms of ``hidden states'' of this second party~\cite{wisemanPRL2007,jonesPRA2007}. Focusing on steering allows, e.g., entanglement verification~\cite{entverification} and quantum key distribution~\cite{branciard2012} in scenarios where one of the parties (or at least her measurement devices) are not trusted. In recent times, steering has found connections with other fundamental topics in quantum information processing like joint measurability~\cite{PhysRevLett.113.160403,PhysRevLett.113.160402} and the discrimination of physical processes~\cite{piani2014einstein}.

In this paper we generalize the notion of steering and steerability to quantum channels. This generalization is such that standard steering for quantum states can be seen as a special case of steering for quantum channels.

The central objects of our investigation are the extensions of a given quantum channel, i.e., broadcast quantum channels~\cite{yard2011quantum} with one sender and two receivers, such that by disregarding one of the outputs we recover the original channel of interest. The practical scenario that one may have in mind is the following. There is a quantum transformation (a quantum channel) from $C$ to $B$, which we can imagine is applied/used by Bob. Such transformation is in general noisy (i.e., not unitary) with, information ``leaking'' to the environment. Suppose Alice has access to some part $A$ of said environment. The question is: is Alice coherently connected to the input-output of the channel, or can she be effectively considered just a ``classical bystander'', with at most access to classical information about the transformation that affected the input of the channel? In the latter case we will think of the map from $C$ to $AB$ as of an \emph{incoherent} extension of the channel from $C$ to $B$.

We will define steerability of a channel extension as the possibility for Alice to prove to Bob that she is not a classical bystander, in the sense that the leakage of  information from $C$ to $A$ cannot be described in terms of a classical channel. The way Alice proves to Bob that she is a ``quantum bystander'', rather than a classical one, is by informing him of the choice of measurements performed on the output $A$, and of the outcomes of such measurements. We will focus on the case where such information only refers to labels for  measurements and outcomes; it is device-independent (just on Alice's side, hence \emph{semi}-device-independent) information, as it does not rely on the details/implementation of such measurements. In this way, the verification procedure does not require Bob (or for that, Alice) to trust Alice's measurement devices. Every choice of measurement by Alice corresponds to a different decomposition into subchannels of the channel used by Bob. We suppose Bob to be able to perform full channel tomography, i.e., have a perfect description of the subchannels  (see the next section for definitions) induced by the measurements of Alice. Based on the reconstructed subchannels, Bob has to decide whether Alice has access to a coherent extension of his channel. If Bob can be convinced that Alice does have access to the other output port of a coherent extension, we say that such an extension is steerable.

In the following, after introducing the notation and the various concepts more formally, we provide several results and make a number of considerations, especially about the similarities and differences between state steering and channel steering. In particular, we argue how, while state steering deals with the notion of hidden states, in the case of channel extensions the approach based on the notion of coherence of the  extension is more proper than the notion of ``hidden channel''. This distinction vanishes in the special---and standard---case of state steering. Other issues that we consider include: how incoherent extensions have an interpretation as extensions realized by Local Operations and Classical Communication (LOCC); the quantification of the steerability of channel extensions; the use of the Choi-Jamio{\l}kowski isomorphism for mapping the study of channel steerability to state steerability.
 
\section{Notation and Definitions}

In this section we set the notation and provide the definitions that will be used in the rest of the paper.

\begin{definition}
Let $\Lambda^{C\rightarrow B}$ be a channel, i.e., a completely-positive trace-preserving linear map. We say that the broadcast channel~\cite{yard2011quantum} $\Lambda^{C\rightarrow AB}$ is a \emph{channel extension} for $\Lambda^{C\rightarrow B}$ if $\Lambda^{C\rightarrow B}=\Tr_A \circ \Lambda^{C\rightarrow AB}$. We say that  $\Lambda^{C\rightarrow AB}$ is a \emph{dilation} in the special case it is an isometric transformation.
\end{definition}

It is well known that the dilation of a channel is uniquely defined up to an isometry on $A$~\cite{nielsen2010quantum}.

\begin{definition}
The collection of completely-positive maps $\{\Lambda_a\}$ is an \emph{instrument} $\cI$ if $\sum_a\Lambda_a$ is a channel. In such a case, each $\Lambda_a$ is a \emph{subchannel}, i.e., a completely positive trace-non-increasing linear map.
\end{definition}
It is a direct consequence of the Stinespring dilation theorem that, given any instrument decomposition of some channel, the information about which subchannel of the instrument is actually applied to the input of the channel is in principle available to some party~\cite{davies1970operational,nielsen2010quantum,michalhorodeckientmeasures}. Indeed, any instrument can be realized by performing a measurement, described by a Positive-Operator-Valued Measure (POVM) $\{M_a\}_a$, $M_a\geq 0$, $\sum_a M_a = \openone$, on the extension $A$ of a dilation of the channel. A proof can also be obtained by considering the results of~\cite{hughston1993complete} that regard the realization of arbitrary ensembles for a given mixed state when one has access to its purification, and by applying such results to the Choi-Jamio{\l}kowski state isomorphic to the channel (see below). On the other hand, any measurement on $A$, for any extension $\channel{C}{AB}$, realizes some instrument for the channel $\channel{C}{B}$.

We now formalize the notion of incoherent extension of a channel.
\begin{definition}
\label{def:incoherent}
We say that $\channel{C}{AB}$ is an \emph{incoherent extension} of $\channel{C}{B}$ if there exists an instrument $\{{\Lambda}_{\lambda}\}_\lambda$ and normalized (i.e., with unit trace) quantum states $\{\sigma_\lambda^A\}_\lambda$ such that
\beq
\label{eq:incoherent}
\channel{C}{AB}=\sum_\lambda \channel{C}{B}_\lambda\otimes \sigma^A_\lambda.
\eeq
We say that $\channel{C}{AB}$ is a \emph{coherent extension} of $\channel{C}{B}$ if it is not incoherent.
\end{definition}
Notice that an incoherent extension can be thought of as the combination of an incoherent extension
\beq
\label{eq:instrument}
\sum_\lambda \channel{C}{B}_\lambda\otimes \proj{\lambda}^{A'},
\eeq
with orthonormal states $\ket{\lambda}$, followed by the preparation $\proj{\lambda}^{A'}\mapsto\sigma_\lambda^A$. 
Definition~\ref{def:incoherent} captures the idea that, for an incoherent extension, the information that leaks to $A$ is at most information about which subchannel $\channel{C}{B}_\lambda$ of the channel $\channel{C}{B}$ was applied to the input (``most'' corresponding to the case where the $\sigma_\lambda$'s are also orthogonal).

Notice also that for any given input $\rho^C$ on $C$, the action of an incoherent extension produces a \emph{separable} (i.e., unentangled) state on $AB$:
\[
\begin{split}
\sum_\lambda \channel{C}{B}_\lambda[\rho^C]\otimes \sigma^A_\lambda = \sum_\lambda p_\lambda\rho^B_\lambda \otimes \sigma^A_\lambda
\end{split}
\]
with $p_\lambda=\Tr_B(\channel{C}{B}_\lambda[\rho^C])$ and $\rho^B_\lambda=\channel{C}{B}_\lambda[\rho^C]/p_\lambda$, and on the right-hand side the generic expression of a separable state.

Finally, notice that we only ask the $\Lambda_\lambda$'s to form an instrument, each being a subchannel, and we do not require each of them to be proportional to a channel. We  return to this point in Section~\ref{sec:channelssuchannels}.

\begin{definition}
\label{def:channelassemblage}
A \emph{channel assemblage} $\cC\cA=\{\cI_x\}_x=\{\Lambda_{a|x}\}_{a,x}$ for a channel $\Lambda$ is a collection of instruments $\cI_x$ for $\Lambda$, i.e., it holds $\sum_a\Lambda_{a|x}=\Lambda$ for all $x$.
\end{definition}

\begin{definition}
\label{def:steerable}
We say that a channel assemblage $\cC\cA=\{\Lambda_{a|x}\}_{a,x}$ is \emph{unsteerable} if there exists an instrument $\{{\Lambda}_{\lambda}\}_\lambda$ and conditional probabilities distributions $p(a|x,\lambda)$ such that
\beq
\Lambda_{a|x} = \sum_\lambda p(a|x,\lambda) {\Lambda}_\lambda.
\eeq
\end{definition}
An unsteerable channel assemblage is such that all instruments that are part of the channel assemblage arise from the action of the same instrument $\{\Lambda_\lambda\}_\lambda$; each instrument $\{\Lambda_{a|x}\}_a$ in the unsteerable channel assemblage corresponds then just  to a classical processing of the information about which $\Lambda_\lambda$ was applied, potentially different for each $x$.

Notice that any measurement assemblage $\cM\cA=\{M^A_{a|x}\}_{a,x}$ on $A$ , with $\{M^A_{a|x}\}_{a}$ a POVM for every $x$, leads to a channel assemblage for $\channel{C}{B}$, via
\beq
\Lambda_{a|x}^{C\rightarrow B}[X]=\Tr_A(M^A_{a|x}\Lambda^{C\rightarrow AB}[X]).
\eeq

\begin{definition}
We say that a channel extension $\channel{C}{AB}$ is \emph{unsteerable} if measurement assemblages on $A$ only lead to unsteerable channel assemblages.
\end{definition}

\begin{definition}
The Choi-Jamio{\l}kowski~\cite{jamiolkowski1972linear,choi1975completely} isomorphic operator to a channel $\channel{C}{B}$ is the operator
\beq
J_{C'B}(\channel{C}{B}):=\channel{C}{B}[\psi_+^{CC'}],
\eeq
where $\psi_+^{CC'}$ is the density matrix corresponding to a fixed maximally entangled state of systems $C$ and $C'$, with $C'$ a copy of $C$.
\end{definition}

The definitions above comprise and generalize the more familiar corresponding notions for quantum states. Indeed, to fall back to the latter case, it is sufficient to consider the case of a channel with constant input, i.e., $\Lambda[\,\cdot\,]=\Tr(\,\cdot\,)\hat{\sigma}$, which we can identify with the (normalized!) state $\hat{\sigma}$. It is easy to see that there is the following mapping to such a special case:
\begin{itemize}
\item channel extension $\channel{C}{AB}$ of a channel $\channel{C}{B}$ $\mapsto$ extension $\rho^{AB}$ of a quantum state $\rho^B$;
\item dilation $V^{C \rightarrow AB}\,\cdot\,V^{C\rightarrow AB, \dagger}$ of a channel $\channel{C}{B}$ $\mapsto$ purification $\psi^{AB}$ of a quantum state $\rho^B$;
\item incoherent extension of Eq.~\eqref{eq:incoherent} $\mapsto$ separable extension $\sum_\lambda p_\lambda \rho_\lambda^A \otimes \rho^B_\lambda$ of a quantum state $\rho^B$;
\item coherent, that is, not incoherent, extension $\mapsto$ entangled, that is, not separable, extension of a quantum state;
\item instrument $\{\Lambda_a\}$ for a channel $\Lambda=\sum_a\Lambda_a$ $\mapsto$ ensemble $\{\rho_a\}$, with $\rho_a$'s subnormalized states, for a quantum state $\rho=\sum_a\rho_a$;
\item channel assemblage $\{\Lambda_{a|x}\}_{a,x}$ $\mapsto$ state assemblage $\{\rho_{a|x}\}_{a,x}$;
\item (un)steerable channel assemblage $\mapsto$ (un)steerable state assemblage.
\item (un)steerable channel extension $\mapsto$ (un)steerable state extension.
\end{itemize}
We remind the reader that a state assemblage for the state $\rho$ is a collection of ensembles $\{\rho_{a|x}\}_{a,x}$, one ensemble per value of $x$, such that $\sum_a \rho_{a|x}=\rho$ for all $x$. This corresponds to Definition~\ref{eq:incoherent} in the case $\Lambda[\,\cdot\,]= \Tr(\,\cdot\,) \rho$ and $\Lambda_{a|x}=\Tr(X) \rho_{a|x}$. An unsteerable assemblage for $\rho$ is one that can be written as $\rho_{a|x}=\sum_\lambda p(a|x,\lambda)\rho_\lambda$, with $\rho_\lambda$ subnormalized and such that $\sum_\lambda\rho_\lambda=\rho$ (compare Definition~\ref{def:steerable}). Finally, an unsteerable state extension $\rho^{AB}$ is such that, upon considering measurement assemblages on $A$, it only gives raise to unsteerable assemblages for $\rho^B$.

\section{Steerability of extension and coherence of extensions}

In the following we relate the coherence of channels extensions to their steerability, in the same way in which one relates the entanglement of an extension to the steerability allowed by such an extension.

\begin{proposition}
Every incoherent  channel extension leads to unsteerable channel assemblages upon steering attempts by Alice. Conversely, every unsteerable channel assemblage can be thought as arising from an incoherent channel extension.
\end{proposition}
\begin{proof} The proof of the claim is straightforward. Indeed, for any measurement assemblage $\{M_{a|x}\}_{a,x}$, and any incoherent channel extension,
\[
\begin{aligned}
\Lambda_{a|x}^{C\rightarrow B}[\,\cdot\,]
&=\Tr_A(M^A_{a|x}\Lambda^{C\rightarrow AB}[\,\cdot\,])\\
&=\Tr_A(M^A_{a|x}\sum_\lambda \channel{C}{B}_\lambda[\,\cdot\,]\otimes \sigma_\lambda^A)\\
&=\sum_\lambda \channel{C}{B}[\,\cdot\,]_\lambda\Tr_A(M^A_{a|x}\sigma_\lambda^A)\\
&=\sum_\lambda \channel{C}{B}_\lambda[\,\cdot\,] p(a|x,\lambda).
\end{aligned}
\]
On the other hand, fix an unsteerable channel assemblage $\{\Lambda^{\textrm{US}}_{a|x}=\sum_\lambda p(a|x,\lambda)\Lambda_\lambda\}_{a,x}$. This can be obtained through the measurement assemblage $\{M_{a|x}=\sum_\lambda p(a|x,\lambda)\proj{\lambda}\}_{a,x}$ applied to the $A$ output of $\channel{C}{AB}=\sum_\lambda \channel{C}{B}_\lambda\otimes \proj{\lambda}^A$. Indeed, 
\[
\begin{aligned}
&\;\quad\Tr_A(M^A_{a|x}\Lambda^{C\rightarrow AB}[\,\cdot\,])\\
&=\Tr_A\Bigg[\left(\sum_{\lambda'}p(a|x,\lambda')\proj{\lambda'}^A\right)\\
&\qquad\quad\times\left(\sum_\lambda \channel{C}{B}_\lambda[\,\cdot\,]\otimes \proj{\lambda}^A\right)\Bigg]\\
&=\sum_{\lambda,\lambda'} \delta_{\lambda,\lambda'}p(a|x,\lambda')\channel{C}{B}_\lambda[\,\cdot\,]\\
&=\sum_\lambda p(a|x,\lambda)\channel{C}{B}_\lambda\\
&=\Lambda^{\textrm{US}}_{a|x}.
\end{aligned}
\] 
\end{proof}
Hence, if Bob verifies that Alice is able to steer his channel, it means that she has access to a coherent extension of the channel. This is a generalization of the fact that the ability of Alice to steer the quantum state of Bob implies that she has access to an entangled extension of the state held by Bob. On the other hand, an unsteerable channel assemblage is always compatible with a incoherent extension. This, in turn is a generalization of the fact that an unsteerable state assemblage is always compatible with Alice holding only a separable extension of the state held by Bob. Thus, channel steering is necessary and sufficient to establish the coherence of the channel extension in a semi-device-independent way, in the same way in which state steering is necessary and sufficient to establish the coherence (that is, the entanglement) of a state extension, in a semi-device-independent way~\cite{wisemanPRL2007,entverification}.

\section{Steerability of extensions and Choi-Jamio{\l}kowski isomorphism}
\label{sec:jamiolkowski}

The next result relates the steerability properties of a channel extension to the steerability properties of its Choi-Jamio{\l}kowski isomorphic state.

\begin{theorem}
\label{thm:choijamiolkowski}
The channel extension $\channel{C}{AB}$ of the channel $\channel{C}{B}$ is steerable if and only if $J_{C'AB}(\channel{C}{AB})$ is steerable by Alice acting on $A$.
\end{theorem}
\begin{proof}
Since steerability is defined as a violation of unsteerability, the statement is equivalent to the fact that the channel extension is unsteerable if and only if its Choi-Jamio{\l}koski isomorphic state is unsteerable. The latter statement is straightforwardly checked using the properties of the isomorphism, in particular its linearity. Indeed, in one direction linearity implies that, for any measurement assemblage $\{M^{A}_{a|x}\}_{a,x}$,
\begin{multline*}
\begin{aligned}
\left(J_{C'AB}(\Lambda^{C\rightarrow AB})\right)^{C'B}_{a|x}:&=\Tr_A(M^A_{a|x}J_{C'AB}(\Lambda^{C\rightarrow AB}))\\
&=J_{C'B}(\Tr_A(M^A_{a|x}\Lambda^{C\rightarrow AB}[\,\cdot\,]))\\
&=J_{C'B}(\Lambda_{a|x}^{C\rightarrow B}).
\end{aligned}
\end{multline*}
Now, if  $\Lambda^{C\rightarrow AB}$ is an unsteerable extension, then the $\Lambda_{a|x}^{C\rightarrow B}$'s form an unsteerable channel assemblage, so that
\[
J_{C'B}(\Lambda_{a|x}^{C\rightarrow B})=\sum_\lambda p(a|x,\lambda) J_{C'B}(\Lambda^{C\rightarrow B}_\lambda),
\]
proving that $J_{C'AB}(\Lambda^{C\rightarrow AB})$ is unsteerable. The fact that if $J_{C'AB}(\channel{C}{AB})$ is unsteerable by Alice then the extension $\channel{C}{AB}$ is unsteerable, can be similarly proven considering the  inverse of the isomorphism.
\end{proof}

The just proven theorem shows that the study of the steerability of channel extensions can be reduced to the study of the steerability of states, in the sense that is sufficient to study the steerability of the Choi-Jamio{\l}koski state isomorphic to the extension. This is not surprising, as the Choi-Jamio{\l}kowski state isomorphic to a linear map encodes all information about that map. Actually, even the fact that a  channel extension is incoherent is encoded in its isomorphic state.
\begin{theorem}
The channel extension $\channel{C}{AB}$ of the channel $\channel{C}{B}$ is incoherent if and only if $J_{C'AB}(\channel{C}{AB})$ is $A:BC'$-separable.
\end{theorem}
\begin{proof}
The ``only if'' implication is obvious. The proof of the ``if'' is also straightforward, once one makes use of the  one-to-one correspondence between operators and maps, and of linearity (see also~\cite{horodecki2003entanglement}).
\end{proof}

It is clear that any non-trivial dilation extension is coherent and steerable, since it corresponds to a pure $J_{C'AB}$ that is $A:BC'$ entangled, and hence steerable by Alice~\cite{wisemanPRL2007,quantifyingsteering}.

\section{Complementary channel for given extension}
In quantum information it is customary and useful to consider the \emph{complementary channel}~\cite{wilde2013quantum} $\channel{C}{A}$ of a channel $\channel{C}{B}$, defined uniquely, up to a isometry at the output, by the consideration of the dilation $\channel{C}{AB}[\,\cdot\,]=V^{C\rightarrow AB}\,\cdot\, V^{C\rightarrow AB,\dagger}$, via $\channel{C}{A}=\Tr_B\circ \channel{C}{AB}$. It should be clear that we can use this latter expression to introduce the notion of complementary channel with respect to a general channel extension, rather than with respect to the channel dilation.

We observe that the complementary channel associated to an incoherent extension is entanglement breaking~\cite{horodecki2003entanglement}, since
\[
\begin{split}
\Tr_B(\sum_\lambda \channel{C}{B}_\lambda [\,\cdot\,]\otimes\sigma^A_\lambda)&=\sum_\lambda \Tr_C((\Lambda_\lambda^{C\rightarrow B})^\dagger[\openone^B] \, \cdot \, )\sigma^A_\lambda\\
&=\sum_\lambda \Tr_C(M^C_\lambda \, \cdot \, )\sigma^A_\lambda,
\end{split}
\]
where $(\Lambda_\lambda^{C\rightarrow B})^\dagger$ is the dual map of $\Lambda_\lambda^{C\rightarrow B}$ and the operators $M^C_\lambda:=(\Lambda_\lambda^{C\rightarrow B})^\dagger[\openone^B] $ form a POVM on $C$.

We emphasize that the implication in the opposite direction does not hold: even if the generalized complementary channel is entanglement breaking, the corresponding extension does not need to  be incoherent. What holds true, as we have seen in Theorem~\ref{thm:choijamiolkowski}, is that a an extension $\channel{C}{AB}$ is incoherent if and only if its Choi-Jamio{\l}kowski isomorphorphic operator $J_{C'AB}(\channel{C}{AB})$ is $A:BC'$-separable. On the other hand, the fact that the complementary channel is entanglement breaking is equivalent to $\Tr_B(J_{C'AB}(\channel{C}{AB}))$  being $A:C'$-separable. Thus, via the isomorphism, any example---up to local filtering---of tripartite state for $ABC'$  that is not $A:BC'$-separable but is such that its reduction on $A:C'$ is separable provides an example of extension that is not incoherent but such that the generalized complementary channel is entanglement breaking. This is interesting because it implies that, even though $A$ by itself may only have access to classical information about the input $C$, she may still exhibit some quantum control on the channel from $C$ to $B$.

\section{Steerability of extensions, channel tomography, and ancilla-assisted channel tomography}

As made clear with Theorem~\ref{thm:choijamiolkowski}, Bob can faithfully test the steerability of an extension by the use of just one entangled input. From the physical point of view, this is the correspondent of ancilla-assisted channel tomoghaphy~\cite{altepeter2003ancilla}, which becomes in this case ancilla-assisted subchannel tomography. Bob may elect not to use entanglement, or be unable to. In such a case, in order to realize the tomography of the subchannels, he needs to use several inputs. Notice nonetheless that, depending on the extension, channel tomography may not be required to establish the steerability of the extension. Indeed, if there is, e.g., a single input $\rho^C$ on $C$ alone (rather than on $C$ and the ancillary system $C'$) such that the state on $AB$ given by $\channel{C}{AB}[\rho^C]$ is steerable by Alice, hence necessarily entangled, then the coherence of the channel extension is already verified. What we want to emphasize, though, is that it might happen that $\channel{C}{AB}[\rho^C]$ is separable, or even factorized, for any input $\rho^C$, even though the extension is coherent. Such is the case, for example, for the trivial channel fixed output channel  $\channel{C}{B}[\,\cdot\,]=\Tr(\,\cdot\,)\hat{\sigma}^B$ that we have already encountered, and its extension $\channel{C}{AB}[X^C]=X^A\otimes \hat{\sigma}^B$. In this case the complementary channel $\Tr_B\circ  \channel{C}{AB} =\id^{C\rightarrow A}$ is the identity channel, i.e. perfectly coherent. In such a case, it is never possible to observe state steering (or even correlations at all!) between the two output spaces $A$ and $B$. Nonetheless the coherence of the extension can be revealed by using non-orthogonal inputs. The easiest way of seeing this is by considering how the use of several inputs is completely equivalent to the use of a (maximally) entangled input of $C$ and $C'$ with the preparation of inputs on $C$ by means of measurements on $C'$. Indeed, such measurements could be the same that, together with the steering on Alice's side, would allow us to conclude that $J_{C'AB}(\Lambda^{C\rightarrow AB})$ was $A:BC'$ entangled.

\section{Quantifying the steerability of channel extensions}

Consider any measure $S_{X\rightarrow Y}$ of $X\rightarrow Y$-steerability on states $\rho^{XY}$. We can define an induced measure of steerability for the extension $\channel{C}{AB}$ of $\channel{C}{B}$ as
\beq
S_{A\rightarrow (C\rightarrow B)}(\channel{C}{AB}):=\sup_{\rho_{CD}}S_{A\rightarrow BD}(\channel{C}{AB}[\rho_{CD}]),
\eeq
where the supremum is over all inputs $\rho_{CD}$, where $D$ is an arbitrary ancillary system.
The quantification of steering has received focused attention only recently, and not many steering quantifiers have been defined and analyzed yet~\cite{quantifyingsteering,piani2014einstein,gallego2014resource}. A steering quantifier for states often gets defined by first defining a steering quantifier for assemblages, and then ``lifting'' the latter to a steering quantifier for states optimizing over all possible steering measurement assemblages. That is,
\[
S_{X\rightarrow Y}(\rho^{XY})=\sup_{\{M^X_{a|x}\}_{a,x}} S(\{\Tr_X(M^X_{a|x} \rho^{XY})\}_{a,x}),
\]
with $S(\{\rho_{a|x}\}_{a,x})$ the steering quantifier for the assemblage $\{\rho_{a|x}\}_{a,x}$.

If  $S_{X\rightarrow Y}$ is convex and invariant under isometries on $Y$, as reasonable~\cite{gallego2014resource}, then,
\begin{multline*}
\sup_{\rho_{CD}}S_{A\rightarrow BD}(\channel{C}{AB}[\rho_{CD}])\\
=\max_{\psi_{CC'}} S_{A\rightarrow BC'}(\channel{C}{AB}[\psi_{CC'}]),
\end{multline*}
because one can focus on pure input states, and by virtue of the Schmidt decomposition $D$ can be taken to be a copy $C'$ of $C$.

It can be argued that a steering quantifier $S_{X\rightarrow Y}$ should not increase under one-way LOCC (a subclass of LOCC operations realized by classical communication in one direction only) from $Y$ to $X$~\cite{gallego2014resource,pianiprep}. Such a class includes obviously local operations, in particular local operations on $X$. If a steering quantifier has such a property, then
\[
S_{A\rightarrow (C\rightarrow B)}(\channel{C}{AB}) \geq S_{A\rightarrow (C\rightarrow B)}(\Gamma^A\circ\channel{C}{AB}),
\]
This can be intuitively be understood, since the local operation $\Gamma^A$ effectively restricts the class of measurements that can be applied to the output $A$ of $\channel{C}{AB}$. Thus, we find that channel extensions that are mapped one into the other by means of an additional local channel on $A$ respect a steering hierarchy at the quantitative level, with dilations having always the highest degree of steerability. This parallels the fact that purifications have the highest degree of steerability in the case of state extensions~\cite{pianiprep}.

Specific choices for $S_{X\rightarrow Y}$ which are efficiently computable at least for fixed assemblages include the steerable weight~\cite{quantifyingsteering} and the robustness of steering~\cite{piani2014einstein}. Adopting them in our case, one can immediately define the steerable weight of extensions and the steering robustness of extensions.

Notice that, in the same way in which one focuses on specific state assemblages, one can also focus on specific channel assemblages, and define steering measures for those, via
\[
\begin{aligned}
S(\{\Lambda^{C\rightarrow B}_{a|x}\}_{a,x})
&= \sup_{\rho_{CD}} S(\{\Lambda^{C\rightarrow B}_{a|x}[\rho_{CD}]\}_{a,x})\\
&= \max_{\psi_{CC'}} S(\{\Lambda^{C\rightarrow B}_{a|x}[\psi_{CC'}]\}_{a,x}),
\end{aligned}
\]
with the second equality valid under the conditions mentioned above.

The detailed quantitative study of the steerability of channel assemblages and channel extensions will be considered elsewhere~\cite{pianiprep}.

\section{Coherence of channel extensions and local operations and classical communication}
\label{sec:LOCC}

All channel extensions can be thought of as originating from a given dilation, followed by partial trace on some part of the environment/extension. That is, each channel extension is the result of Alice having access to only a certain part of the environment to which quantum information has leaked to. Depending on what kind of information remains available in the subsystem $A$ of the environment to which Alice has access to, the channel extension $\channel{C}{AB}$ may be coherent or not.

There is another way of thinking of channel extensions though, and it is by imagining how the extension can be \emph{implemented}. The ``how'' refers here to the operations and/or resources at disposal to realize, we could say, the inside of the (quantum) ``box'' $\channel{C}{AB}$. Indeed, this is the view behind Definition~\ref{def:incoherent} of incoherent channel extension: it is an extension that requires only one-way LOCC to be implemented, while in general an extension would require quantum communication and/or global operations.

We argue now that the incoherent channel extensions of Definition~\ref{def:incoherent} are the most general extensions that can be realized by \emph{full} LOCC. The reason is simple. General LOCC operations differ from one-way LOCC operations because of the possibility of two-way classical communication. The question is whether ``backward'' classical communication from the environment $A$ can enlarge the set of channel extensions. The answer is negative, for a simple reason. The system $A$ is initially uncorrelated with the input $C$. That means that whatever the communication from $A$ at any point in the LOCC protocol, it does not provide any useful information on how to process the input $C$, and can be ``simulated'' locally. Only classical communication from $C$ to $A$ is relevant, and can be considered to happen after local operations (described by the subchannels) on $C$, mapping it to $B$, have been completed. Hence LOCC extensions are the same as one-way LOCC extensions. Notice that the same reasoning applies to general extensions that require quantum communication: only quantum communication from $C$ to $A$ is needed at most to realize the most general (i.e., possibly coherent) extension.

In summary, the incoherent extensions of Definition~\ref{def:incoherent} correspond to the most general extensions that can be realized classically, that is, via LOCC, or, equivalently, without quantum communication or shared entanglement.

\section{Hidden states and hidden channels}
\label{sec:channelssuchannels}

In this section we discuss the fact that the definition of incoherent extensions, see Definition~\ref{def:incoherent}, involves subchannels, rather than rescaled (by means of a probability distribution) channels.

An alternative definition of incoherent operations could have been
\beq
\label{eq:incoherentchannel}
\channel{C}{AB}=\sum_\lambda p_\lambda \hat{\Lambda}^{C\rightarrow B}_\lambda\otimes \sigma^A_\lambda,
\eeq
where we have emphasized by a $\hat{}$ (hat) that the maps $\hat{\Lambda}^{C\rightarrow B}_\lambda$ are in this case channels themselves, and $\{p_\lambda\}_\lambda$ is a probability distribution. In such a case one can imagine that one random channel among the $\hat{\Lambda}^{C\rightarrow B}_\lambda$'s is applied to the input, with a specific probability for each channel to take place. Thus, the convex combination that leads to the channel $\channel{C}{B}=\Tr_A\circ \channel{C}{AB}=\sum_\lambda p_\lambda \hat{\Lambda}^{C\rightarrow B}_\lambda$ can be thought as the result of \emph{ignorance} about which actual channel $\hat{\Lambda}^{C\rightarrow B}_\lambda$ \emph{is going} to be applied, but, at least in principle, someone---for example whoever has access to $A$, in the case the states $\sigma^A_\lambda$ are orthogonal---would be able to know this information \emph{in advance}. Thus, in this case one can really think of ``hidden channels''. These would be hidden to Bob, but not to some other party, e.g., Alice, in the same way in which Alice may know about hidden states of the system held by Bob in the case of an unsteerable state assemblage.

We emphasize that definitions~\eqref{eq:incoherent} and~\eqref{eq:incoherentchannel} coincide when we fall back to the case of state extensions, rather than general channel extensions. Indeed, incoherent extensions in the sense of Eq.~\eqref{eq:incoherent}, when we specialize to states  have subchannels $\Lambda_\lambda^{C\rightarrow B}[\,\cdot\,]=\Tr(\,\cdot\,)\sigma_\lambda= p_\lambda \Tr(\,\cdot\,)\hat{\sigma}_\lambda=p_\lambda \hat{\Lambda}_\lambda(\,\cdot\,)$, with  $\{p_\lambda=\Tr(\sigma_\lambda)\}_\lambda$ a probability distribution, $\hat{\sigma}_\lambda=\sigma_\lambda/p_\lambda$ states, and each $\hat{\Lambda}(\,\cdot\,)$ a channel. Thus, we recover Eq.~\eqref{eq:incoherentchannel}.

In the case of the definition of incoherent channel extensions that we decided to adopt, Definition~\ref{def:incoherent}, one can of course discuss the notion of ``hidden subchannels''. What is key, though, is that, even if someone has knowledge of the specific \emph{set of subchannels}, if the latter are generic subchannels---rather than trace-rescaling suchannels---the ``which subchannel'' information is not available in advance to anyone: which subchannel gets actually applied necessarily depends also on the input to the channel.

If we assume the perspective that the notion of incoherent channel extension should be able to incorporate the maximal classical information available to a classical bystander about the transformation of the input, rather than information about how an arbitrary input is going to be transformed \emph{in the future}, we thus see Definition~\ref{def:incoherent} as the correct one. Other considerations support this point of view: \begin{itemize}
\item this definition captures the notion of most general extension realized by LOCC, as discussed in Section~\ref{sec:LOCC};
\item this definition is consistent with the special case of standard incoherent state extensions (separable states), as just discussed in this section;
\item allows a direct, simple and elegant correspondence between incoherent channel extensions and incoherent state extensions via the Choi-Jamio{\l}kowski isomorphism, as discussed in Section~\ref{sec:jamiolkowski}.
\end{itemize}

\subsection{Incoherent extensions of extremal channels}

It is worth reminding the reader that the structure of the convex set of channels is far from trivial. A channel from $C$ to $B$ admits \emph{Kraus decompositions} of the form
\beq
\label{eq:kraus}
\channel{C}{B}[\,\cdot \,]=\sum_iK_i \, \cdot \,K^{ \dagger}_i,
\eeq
with the \emph{Kraus operators} $K_i=K_i^{C\rightarrow B}$, and $\sum_i K_i^\dagger K_i=\openone^C$.
The extremal channels admit a decomposition~\eqref{eq:kraus} with Kraus operators such that the set of matrices $\{K_i^\dagger K_j\}_{i,j}$ is linearly independent. So, for example, there are extremal channels with more than one Kraus operator that are extremal.

If we were to adopt the definition~\eqref{eq:incoherentchannel} for incoherent extensions, the only allowed incoherent extensions for an extremal channel $\channel{C}{B}$ would be trivial, that is, of the form $\channel{C}{B}[\,\cdots\,]=\channel{C}{B}\otimes \sigma^A$. On the other hand, it is clear that more details about the transformation can be acquired by a classical Alice; for example, information about which Kraus operator got applied, via the channel extension
\[
\sum_i K_i\,\cdot\,K_i^\dagger\otimes \proj{i}^A.
\]

\section{Broadcast channels}

We have considered extensions of a fixed channel. We already mentioned that such extensions are also called \emph{broadcast channels}, i.e., channels with one sender (Charlie) and two receivers (Alice and Bob).

If one considers a fixed broadcast channel $\channel{C}{AB}$, it is interesting to consider both \emph{channel reductions},
\[
\begin{aligned}
\channel{C}{B} &= \Tr_A\circ \channel{C}{AB},\\
\channel{C}{A} &= \Tr_B\circ \channel{C}{AB},
\end{aligned}
\]
which are each the generalized complementary channel of the other with respect to the broadcast channel $\channel{C}{AB}$.

From this perspective, it natural to study how much each receiver can steer the channel between the sender and the other receiver. Notice that all the considerations and results we discussed so far remain perfectly valid. So, for example, we can focus on the steerability and coherence properties of $J_{C'AB}(\channel {C}{AB})$ in the $A:BC'$ and $B:AC'$ bipartite cuts to study the steerability and coherence properties of two the channel reductions (with respect to the given extension).

Note that we are now generalizing the conceptual approach to the study of,steering where one does not look at a bipartite extension of a fixed reduced state, but is rather interested in the steering properties of a bipartite state, i.e., in the extent to which one party can steer the other party. 

Applications of channel steering to the problem of broadcast quantum communication, in particular to the case where the receivers are allowed to communicate classically and cooperate, is left open for future research.

\section{Conclusions}

We have generalized the notion of quantum steering from property of quantum states to property of quantum channels. We have discussed differences and similarities between the case of states and the case of channels, focusing in particular on the issue of whether it makes sense to say that channel steering rules out the possibility of hidden channels, rather than the possibility of an incoherent extension. We also discussed in the detail how the notion for incoherence of extensions that we have adopted is well motivated operationally, since it comprises the most general LOCC implementation of a channel extension. Quite importantly, we have shown both a qualitative and a quantitative connection between state steering and channels steering. On one hand, the Choi-Jamio{\l}kowski isomorphism allows us to map the study of channel steerability to the study of state steerability, leveraging known results in a new context. On the other hand, we have shown how tools developed to quantify state steering can be readily adopted to quantify channel steering.

Finally, we note how, in the framework offered by quantum information processing,  the status of quantum steering (in the case of  states) has changed significantly: from weird quantum feature, to property that allows the semi-device-independent verification of entanglement. Moving to the analysis of the steering of channels is a natural step, that could see steering become an even more useful tool in quantum information processing, particularly in a multipartite scenario revolving around broadcast channels.

\section*{Ackowlegements}
I would like to thank J. Watrous and D. Leung for discussions. This research has been partially supported by NSERC and CIFAR.

\end{document}